%% file: main.tex
\documentclass[11pt]{article}
\usepackage[a4paper,margin=26mm]{geometry}
\usepackage[T1]{fontenc}
\usepackage{lmodern}
\usepackage{amsmath,amssymb,amsthm,mathtools,booktabs,array,enumitem,microtype,hyperref}
\hypersetup{colorlinks=true,linkcolor=blue,citecolor=blue,urlcolor=blue,
  pdftitle={A Four-by-Four Obstruction to Polynomial Removal for Ordered Binary Matrices},
  pdfauthor={Zinan Wang},pdfsubject={Ordered binary matrix removal}}
\newtheorem{theorem}{Theorem}[section]
\newtheorem{lemma}[theorem]{Lemma}
\newtheorem{proposition}[theorem]{Proposition}
\newtheorem{corollary}[theorem]{Corollary}
\newtheorem{remark}[theorem]{Remark}
\newcommand{\dist}{\operatorname{dist}}
\newcommand{\Copy}{N_{H_4}}
\newcommand{\ind}{\mathbf1}
\newcommand{\code}[1]{{\small\nolinkurl{#1}}}
\title{A Four-by-Four Obstruction to Polynomial Removal for Ordered Binary Matrices}
\author{Zinan Wang}
\date{October 11, 2026}

\begin{document}
\maketitle

\begin{abstract}
We exhibit a fixed $4\times4$ binary matrix with no polynomial induced
ordered removal bound. At depth $h$, write $m=2^h$. Our explicit host has
order $(14h+2)m$, exactly $2m^6+\frac12(m-2)m^5$ induced copies,
and exact editing distance $m^2$ cells. Consequently its normalized
edit distance is $\Theta(h^{-2})$, while its copy density is
$\Theta(4^{-h}h^{-8})$. The lower bound holds against arbitrary cell
edits by an injective family of $m^2$ protected selections; an
arithmetic classification identifies every copy and certifies that a
matching canonical repair creates none. The original-copy transversal
number is exactly $m$, establishing an unbounded gap between hitting
existing copies and eliminating all copies. The identities and the
all-depth obstruction are formalized in Lean, with complete certificate
sources and reproducible verification procedures.
\end{abstract}

\section{Introduction and main results}
Let $H$ be a binary $k\times k$ matrix and let $A$ be a binary $n\times n$
matrix. An induced ordered copy of $H$ in $A$ is a pair of strictly
increasing $k$-tuples of row and column indices whose induced submatrix
agrees with $H$ in every entry. Write $N_H(A)$ for the number of such
copies. Let $D_H(A)$ be the minimum number of cell changes that turn $A$
into an $H$-free matrix, allowing both $0\to1$ and $1\to0$ changes, and
write $\dist_H(A)=D_H(A)/n^2$.
A polynomial induced removal bound for $H$ is an inequality of the form
\begin{equation}\label{eq:removal-bound}
 N_H(A)\ge c\,\dist_H(A)^C n^{2k}
\end{equation}
with constants $c,C>0$ independent of $A$ and $n$. Ordered removal lemmas
and their property-testing consequences are studied in \cite{ABF,ABE20}.
The earlier construction of \cite{OAI} exhibited a $66\times66$
obstruction to polynomial removal. Here the forbidden pattern is
$4\times4$ and remains fixed at every depth. We determine its exact
copy count and unrestricted editing distance, and identify the original-copy
transversal number. Unlike an original-copy hitting set, an editing set
must leave a matrix free of all copies, including those created by the
edits (see \cite{ABE20}). The smaller forcing frame and its terminal
classification provide the principal new quantitative results.

We consider the pattern
\begin{equation}\label{eq:pattern}
 H_4=\begin{pmatrix}
 0&1&1&0\\
 1&0&0&1\\
 1&0&1&0\\
 0&1&1&1
 \end{pmatrix}
 =\begin{pmatrix}J_2-I_2&I_2\\I_2&P\end{pmatrix},
 \qquad P=\begin{pmatrix}1&0\\1&1\end{pmatrix}.
\end{equation}
The rows of $H_4$ are pairwise distinct, as are its columns. Hence physical
copies can be counted by weighted block tuples (Lemma~\ref{lem:quotient}).

\begin{theorem}[Main theorem]\label{thm:main}
For every integer $h\ge1$, put $m=2^h$ and $n_h=(14h+2)m$. Let $A_h$ be
the binary matrix constructed in Section~\ref{sec:host}. Then
\begin{equation}\label{eq:main}
 D_{H_4}(A_h)=m^2,\qquad
 \dist_{H_4}(A_h)=\frac{1}{(14h+2)^2},\qquad
 \Copy(A_h)=2m^6+\frac{m-2}{2}m^5.
\end{equation}
Moreover, the canonical repair $B_h$ is $H_4$-free, and every induced
ordered copy in $A_h$ belongs to exactly one of the two terminal families
described in Section~\ref{sec:classification}.
\end{theorem}

\begin{corollary}[Failure of polynomial removal]\label{cor:failure}
The pattern $H_4$ admits no polynomial induced ordered binary matrix removal
bound. More precisely, if
\[
 \epsilon_h=(14h+2)^{-2},\qquad
 \delta_h=\frac{\Copy(A_h)}{n_h^8},
\]
then
\[
 n_h\longrightarrow\infty,\qquad \epsilon_h\longrightarrow0,
 \qquad \frac{\delta_h}{\epsilon_h^C}\longrightarrow0
 \quad\text{for every fixed }C\in\mathbb R.
\]
\end{corollary}
\begin{proof}
The theorem gives, with $d_h=14h+2$,
\[
 \delta_h=\frac{5/2-2^{-h}}{d_h^8 4^h}\le\frac{3}{4^h}.
\]
Since $d_h\to\infty$ and $n_h=d_h2^h$, we have $n_h\to\infty$ and
$\epsilon_h\to0$.
Fix $C\in\mathbb R$ and choose a natural number $k>C$. Since
$0<\epsilon_h\le1$, we have
\[
 0\le\frac{\delta_h}{\epsilon_h^C}
 \le\frac{\delta_h}{\epsilon_h^k}
 \le\frac{3d_h^{2k}}{4^h}\longrightarrow0.
\]
The last limit is the standard comparison of polynomial growth with
geometric growth. Given $c,C>0$, the inequality
$\delta_h<c\epsilon_h^C$ therefore holds for all sufficiently large $h$,
which contradicts \eqref{eq:removal-bound} along the family $A_h$.
\end{proof}

The proof separates the unrestricted editing problem from the enumeration
of copies. Section~\ref{sec:host} defines the block-constant host and its
canonical repair. Section~\ref{sec:distance} constructs $m^2$ protected
selections and proves the lower bound for arbitrary edits. The terminal
classification and repaired freeness appear in
Section~\ref{sec:classification}; Section~\ref{sec:count} then closes the
main theorem by counting the physical copies with their block
multiplicities. The exact transversal number and the quantitative density
profile follow in subsequent sections. Structural details, complementary
proofs, arithmetic certificates and verification instructions are collected
in the appendices. The tree-propagation mechanism follows the general
approach of \cite{OAI}; the forcing frame and the exact classifications
here are specific to $H_4$.

\section{The host construction}\label{sec:host}
Fix $h\ge1$ and write $m=2^h$. For each internal level $0\le i<h$ and
node $0\le p<2^i$, introduce row blocks $R_i^\pm(p)$ and column blocks
$C_i^\pm(p)$, each of size $2^{h-i}$. At level $h$ the two signs are
identified: a leaf label $p<m$ determines one row position $R_h(p)$ and
one column position $C_h(p)$. Finally, add an $m$-element dummy block on
each axis, denoted by $R_*$ and $C_*$.

The variable blocks form binary trees. On the row axis the minus block at a
node precedes the two child subtrees and the plus block follows them; on
the column axis these two positions are reversed. In both trees the left
child subtree precedes the right child subtree. The row dummy block follows
the tree, whereas the column dummy block precedes it.

For each $1\le i\le h$ we add six marker modes
\[
 V_i^+,\ V_i^-,\ W_{i,0}^+,\ W_{i,1}^+,
 \ W_{i,0}^-,\ W_{i,1}^-.
\]
Each mode $t$ has marker row blocks $A_1(t),A_2(t)$ and marker column
blocks $B_1(t),B_2(t)$, all of size $m$. The two marker
groups interact with the following non-marker classes:
\begin{center}\small
\begin{tabular}{c|cc|cc}\toprule
Mode&Row role 1&Row role 2&Column role 1&Column role 2\\\midrule
$V_i^+$&$R_i^+$&$R_{i-1}^+$&$C_*$&$C_{i-1}^+$\\
$V_i^-$&$R_{i-1}^-$&$R_i^-$&$C_*$&$C_{i-1}^-$\\
$W_{i,b}^+$&$R_{i,b}^+$&$R_*$&$C_{i-1}^+$&$C_{i,b}^+$\\
$W_{i,b}^-$&$R_{i,b}^-$&$R_*$&$C_{i,b}^-$&$C_{i-1}^-$\\\bottomrule
\end{tabular}
\end{center}
Here $R_i^\pm$ and $C_i^\pm$ denote the unions of all blocks at the
indicated level and sign; a subscript $b$ restricts the union to nodes of
parity $b$. At level $h$, a leaf
belongs to the class determined by its parity for either sign.

All marker blocks precede the variable and dummy blocks. The marker prefix
starts with
\[
 W_{1,0}^-,W_{1,1}^-,\ldots,W_{h,0}^-,W_{h,1}^-;
\]
the remaining modes occur in increasing level order as
$V_i^+,V_i^-,W_{i,0}^+,W_{i,1}^+$. Within each mode, group 1 precedes
group 2. These orders are used separately on the row and column axes.

The entry rule is constant on every product of two blocks. Two marker
blocks have entry one exactly when their modes agree and their group
indices differ. The support of a group-$u$ marker row on non-marker columns
is exactly column role $u$ of its mode; the support of a group-$u$ marker
column on non-marker rows is exactly row role $u$. Every dummy--non-marker
entry is one. The remaining
variable--variable entries are
\begin{align}\label{eq:values}
 A_h(R_i^+(p),C_i^+(q))&=\ind[p\le q] &&(i<h),\nonumber\\
 A_h(R_i^-(p),C_i^-(q))&=\ind[p<q] &&(i<h),\nonumber\\
 A_h(R_h(p),C_h(q))&=\ind[p\le q],\\
 A_h(R_i^+(p),C_{i-1}^+(q))&=\ind[\lfloor p/2\rfloor\le q],\nonumber\\
 A_h(R_i^-(p),C_{i-1}^-(q))&=\ind[\lfloor p/2\rfloor<q].\nonumber
\end{align}
In the last two identities $1\le i\le h$, and a leaf row at $i=h$ is
treated as either signed row. The indicator $\ind[S]$ is one when $S$
holds and zero otherwise.
All other variable--variable entries are zero.

The canonical repair $B_h$ changes precisely the entries in
\[
 R_*\times\{C_h(p):0\le p<m\}
\]
from one to zero and leaves all other entries unchanged.

\begin{lemma}[Physical size and order]\label{lem:size}
Each axis contains exactly $(14h+2)m$ positions. After ordering the
positions within every block by their offsets, the resulting block order
extends to a strict total order. Consequently each axis is order-isomorphic
to $\{0,\ldots,n_h-1\}$ with its usual order.
\end{lemma}
\begin{proof}
The marker blocks contribute $6h\cdot2m=12hm$ positions. At each internal
depth, the $2^i$ nodes on either sign have multiplicity $2^{h-i}$, hence
contribute $m$ positions; the two signs contribute $2hm$ in total. The
leaves and the dummy block contribute $m$ positions each. This gives
$(14h+2)m$. Ordering offsets increasingly inside each block gives the
required total order and the order-preserving bijections. We use these
bijections to regard $A_h$ and $B_h$ as numerically indexed matrices.
\end{proof}

\begin{lemma}[Block quotient and multiplicity]\label{lem:quotient}
An induced ordered copy of $H_4$ in either host uses four distinct row
blocks and four distinct column blocks. The ordered block tuples whose
entry table equals $H_4$ are precisely the block projections of physical
copies. A tuple has exactly the product of its eight block sizes as
physical realizations.
\end{lemma}
\begin{proof}
Two rows from one block agree on every column, so their signatures in a
copy would be equal. The distinct rows of $H_4$ exclude this possibility;
the distinct columns give the same conclusion on the other axis. Each
block is an interval in its axis order. Conversely, choosing one position
from each block of an ordered tuple preserves its order and all entries,
because the entry rules are block-constant. The choices are independent
and have the stated product as their number.
\end{proof}

\section{The edit-distance lower bound}\label{sec:distance}
A \emph{marker frame} is an ordered choice of four rows and four columns
whose twelve entries outside the lower-right $2\times2$ block agree with
$H_4$. Its lower-right block is a copy of $P$ exactly when the frame is a
copy of $H_4$.

\subsection{Protected selections}
Index a family of selections by pairs $(a,b)\in\{0,\ldots,m-1\}^2$.
Set $z=a\mathbin{\oplus}b$, with $\oplus$ denoting bitwise XOR, and at
depth $d\le h$ select the node
\[
 p_d=\left\lfloor\frac{z}{2^{h-d}}\right\rfloor.
\]
Then $p_0=0$, $p_h=z$ and $p_{d-1}=\lfloor p_d/2\rfloor$ for
$1\le d\le h$.
In a signed row block on this path choose the offset
$a\bmod2^{h-d}$; in the corresponding signed column block choose
$b\bmod2^{h-d}$. In each marker block and in the dummy block use offset
$a$ on rows and offset $b$ on columns. At depth $h$, the two signed paths
meet at the same leaf on each axis.

Call a cell in the product of the selected rows and columns protected
when at least one endpoint is a marker or dummy position, or when both
endpoints are root variable positions. Thus every
selected frame contributes its twelve prescribed entries, together with
the dummy incidences and root entries used in the propagation argument.

\begin{lemma}[Injectivity of protection]\label{lem:injective}
Every physical cell is protected by at most one of the $m^2$ selections.
\end{lemma}
\begin{proof}
If both endpoints are marker or dummy positions, their offsets recover
$(a,b)$. Suppose instead that the row endpoint is a variable position at
depth $d$ and the column endpoint is a marker or dummy position. The
column offset recovers $b$. The row node recovers
$\lfloor z/2^{h-d}\rfloor$, while the row offset recovers
$a\bmod2^{h-d}$. Since XOR commutes with truncation of the lower $h-d$
bits,
\[
 \left\lfloor\frac a{2^{h-d}}\right\rfloor
 =p_d\mathbin{\oplus}
 \left\lfloor\frac b{2^{h-d}}\right\rfloor.
\]
The quotient and remainder therefore determine $a$. When the row is a
marker or dummy and the column is variable, the row offset gives $a$,
the column node gives $p_d$, and the identity
$\lfloor b/2^{h-d}\rfloor=p_d\mathbin{\oplus}\lfloor a/2^{h-d}\rfloor$
recovers the quotient of $b$; its column offset supplies the remainder.
A protected cell with two variable
endpoints occurs at the roots, where its two offsets directly recover
$(a,b)$. At a leaf $d=h$, the offset is zero and $p_h=z$, so the mixed
endpoint cases reduce to $a=z\mathbin{\oplus}b$ or
$b=z\mathbin{\oplus}a$, respectively.
\end{proof}

\begin{corollary}[An untouched protected selection]\label{cor:avoid}
If fewer than $m^2$ cells are changed, some selection has all its protected
cells unchanged.
\end{corollary}
\begin{proof}
By Lemma~\ref{lem:injective}, each changed cell affects at most one
selection. Fewer than $m^2$ changed cells can therefore affect fewer than
$m^2$ selections.
\end{proof}

\subsection{Propagation along a protected path}
Fix a selection whose protected cells remain unchanged, and let $D$ be an
$H_4$-free matrix obtained from $A_h$ by editing cells. Let $x_i$ and $y_i$
be the selected plus and minus diagonal entries at depth $i$, respectively.
Thus $x_i=D(r_i^+,c_i^+)$ and $y_i=D(r_i^-,c_i^-)$, where lower-case
letters denote the selected positions in the corresponding blocks. The protected
root entries give $x_0=1$ and $y_0=0$. Put
$\alpha_i=D(r_i^+,c_{i-1}^+)$ and
$\beta_i=D(r_i^-,c_{i-1}^-)$. The four selected marker frames have bodies
\begin{center}
\begin{tabular}{c|c|l}\toprule
Mode&Body&Consequence of $H_4$-freeness\\\midrule
$V_i^+$&$\left(\begin{smallmatrix}1&\alpha_i\\1&x_{i-1}\end{smallmatrix}\right)$
 &$x_{i-1}=1\ \Longrightarrow\ \alpha_i=1$\\[7pt]
$V_i^-$&$\left(\begin{smallmatrix}1&y_{i-1}\\1&\beta_i\end{smallmatrix}\right)$
 &$y_{i-1}=0\ \Longrightarrow\ \beta_i=0$\\[7pt]
$W_{i,p_i\bmod2}^+$&$\left(\begin{smallmatrix}\alpha_i&x_i\\1&1\end{smallmatrix}\right)$
 &$\alpha_i=1\ \Longrightarrow\ x_i=1$\\[7pt]
$W_{i,p_i\bmod2}^-$&$\left(\begin{smallmatrix}y_i&\beta_i\\1&1\end{smallmatrix}\right)$
 &$\beta_i=0\ \Longrightarrow\ y_i=0$\\\bottomrule
\end{tabular}
\end{center}
The ordered non-marker rows and columns of these frames are
\begin{center}\small
\begin{tabular}{c|cc|cc}\toprule
Mode&First row&Second row&First column&Second column\\\midrule
$V_i^+$&$R_i^+(p_i)$&$R_{i-1}^+(p_{i-1})$&$C_*$&$C_{i-1}^+(p_{i-1})$\\
$V_i^-$&$R_{i-1}^-(p_{i-1})$&$R_i^-(p_i)$&$C_*$&$C_{i-1}^-(p_{i-1})$\\
$W_{i,p_i\bmod2}^+$&$R_i^+(p_i)$&$R_*$&$C_{i-1}^+(p_{i-1})$&$C_i^+(p_i)$\\
$W_{i,p_i\bmod2}^-$&$R_i^-(p_i)$&$R_*$&$C_i^-(p_i)$&$C_{i-1}^-(p_{i-1})$\\\bottomrule
\end{tabular}
\end{center}
Each pair is preceded by the two marker groups of its mode. In the row
tree a plus descendant precedes its plus ancestor, whereas a minus ancestor
precedes its minus descendant. In the column tree these relations are
reversed. The column dummy precedes all variables and the row dummy follows
them. These four observations give the strict orders in every row of the
table, including the terminal level. The twelve marker entries and all
displayed dummy entries are protected. Induction on $i$
now gives
\[
 x_{i-1}=1\Longrightarrow\alpha_i=1\Longrightarrow x_i=1,
 \qquad
 y_{i-1}=0\Longrightarrow\beta_i=0\Longrightarrow y_i=0.
\]
Starting from $x_0=1$ and $y_0=0$, these implications give $x_i=1$ and
$y_i=0$ for every $i\le h$. At depth $h$ the signed row positions are
identified, and so are the signed column positions. Thus $x_h$ and $y_h$
are the same physical cell, which would have to equal both $1$ and $0$.

\begin{proposition}[Arbitrary-edit lower bound]\label{prop:lower}
Every binary $H_4$-free matrix obtained from $A_h$ by cell edits differs
from $A_h$ in at least $m^2$ cells.
\end{proposition}
\begin{proof}
If fewer than $m^2$ cells were changed, Corollary~\ref{cor:avoid} would
provide an untouched protected selection. The propagation argument above
would then produce a contradiction.
\end{proof}

\section{Terminal classification and canonical repair}\label{sec:classification}
For a leaf label $p<m$, put $t=W^-_{h,p\bmod2}$. The two terminal forms use the common row
tuple
\[
 A_1(t),\quad A_2(t),\quad R_h(p),\quad R_*.
\]
Their column tuples are
\begin{align*}
 \mathrm{I}:\quad&B_1(t),\ B_2(t),\ C_h(p),\ C^-_{h-1}(\lfloor p/2\rfloor),\\
 \mathrm{II}:\quad&B_1(t),\ C_h(q),\ C_h(p),
 C^-_{h-1}(\lfloor p/2\rfloor),
 \quad q<p,\ q\equiv p\pmod2.
\end{align*}
The row tuple is ordered because the marker groups precede the leaf and
the dummy row follows it. The leaf column $C_h(p)$ precedes its minus
parent in the column tree.
For form II, the conditions $q<p$ and $q\equiv p\pmod2$ imply
$\lfloor q/2\rfloor<\lfloor p/2\rfloor$, so $C_h(q)$ also precedes
that parent and precedes $C_h(p)$. Thus both displayed tuples are ordered.
The entry table gives the four row vectors
\[
 (0,1,1,0),\quad(1,0,0,1),\quad(1,0,1,0),\quad(0,1,1,1).
\]
In form II, the same parity makes $C_h(q)$ a column-role-1 leaf for
$t$, and $q<p$ makes its entry against $R_h(p)$ zero. These give
the same second-column signature as the second marker in form I.

\begin{lemma}[Exhaustive terminal classification]\label{lem:classification}
A valid ordered block tuple forms an induced ordered copy of $H_4$ in
$A_h$ if and only if it has terminal form I or terminal form II.
\end{lemma}
\begin{proof}
Let $r_1<\cdots<r_4$ and $c_1<\cdots<c_4$ be the selected block
types. By Lemma~\ref{lem:encoding}, their integer encodings satisfy the
original copy predicate $\mathcal C_h^0$. The first implication of
Lemma~\ref{lem:arithmetic} places them in the integer terminal family.

The terminal predicate identifies the first three rows as two level-$h$
$W^-$ marker rows followed by a terminal leaf, the fourth row as the dummy,
and the third column as the same terminal leaf. The fourth column is the
minus parent of that leaf. The entry between the leaf row and the first
marker column is $(H_4)_{31}=1$. By the role table it forces the marker
parity to equal $p\bmod2$. The second column is
then either the second marker of the same mode or an earlier terminal leaf
with the same parity. Physical block validity restores $0\le p<m$;
in the second alternative it also gives $0\le q<p$. These are precisely
forms I and II. The leaf row determines $p$, and the kind of the second
column distinguishes the alternatives. In form II that column also
determines $q$, proving uniqueness of the parameters.

The ordered tuples and their entries were checked above, giving the
converse.
\end{proof}

\begin{lemma}[Freeness of the canonical repair]\label{lem:repair}
For every $h\ge1$, the repaired matrix $B_h$ is $H_4$-free.
\end{lemma}
\begin{proof}
Suppose a copy exists in $B_h$. Lemma~\ref{lem:quotient} gives an ordered
block tuple realizing it, and Lemma~\ref{lem:encoding} sends that tuple
to a witness for $\mathcal C_h^1$. The second implication of
Lemma~\ref{lem:arithmetic} excludes such a witness. This argument applies
to the entire repaired entry table. The bijections in
Lemma~\ref{lem:size} identify the conclusion with freeness in the usual
numerical matrix.
\end{proof}

\section{Exact count and distance}\label{sec:count}
By Lemma~\ref{lem:quotient}, each terminal tuple contributes the product of
its block multiplicities. For form I this product is $2m^5$: the four
marker blocks and the dummy row each have size $m$, the parent column has
size two, and the leaf row and column each have size one. For form II there
are three marker blocks, a dummy row, three leaf positions and a parent
column, giving $2m^4$. The contributions are summarized by
\begin{center}
\begin{tabular}{c|c|c|c}\toprule
Family&Number of parameters&Copies per parameter&Contribution\\\midrule
I&$m$&$2m^5$&$2m^6$\\
II&$m(m-2)/4$&$2m^4$&$(m-2)m^5/2$\\\bottomrule
\end{tabular}
\end{center}

Indeed, the admissible pairs for form II are the strictly ordered pairs
within each of the two parity classes, each of size $m/2$. Hence
\[
 \#\{(q,p):0\le q<p<m,\ q\equiv p\pmod2\}
 =2\binom{m/2}{2}=\frac{m(m-2)}4.
\]
For $h=1$, $m=2$ and the second family is empty, as the formula gives.
The two families are disjoint and their parameters are uniquely recovered
from the selected blocks. The classification lemma therefore yields
\[
 \Copy(A_h)=2m^6+\frac{m-2}{2}m^5.
\]

\begin{proof}[Proof of Theorem~\ref{thm:main}]
Lemma~\ref{lem:size} gives $n_h=(14h+2)m$, and the preceding count gives
the last identity in \eqref{eq:main}. The canonical repair changes exactly
the $m^2$ cells in
$R_*\times\{C_h(p):p<m\}$ and is $H_4$-free. Proposition~\ref{prop:lower}
shows that no repair can use fewer cells. Thus $D_{H_4}(A_h)=m^2$,
and dividing by $n_h^2$ gives the stated normalized distance. The two
remaining assertions are Lemmas~\ref{lem:classification} and
\ref{lem:repair}.
\end{proof}

\section{Hitting original copies versus eliminating all copies}\label{sec:hitting}
\input{additions/exact_hitting}
\input{additions/parameters}
\appendix
\input{additions/context}
\input{additions/structure}
\input{additions/alternate_lower}
\input{additions/robust}
\input{additions/repair_spectrum}
\input{formalization}
\input{additions/compression}
\input{additions/verify}

\end{document}

%% file: additions/exact_hitting.tex
\begin{corollary}[Exact transversal number]\label{cor:exacttransversal}
For every $h\ge1$, with $m=2^h$, the host satisfies
\[
 \tau_{H_4}(A_h)=m<m^2=D_{H_4}(A_h).
\]
In particular, the ratio of unrestricted editing distance to the
minimum number of entries hitting all original copies is exactly $m$.
\end{corollary}
\begin{proof}
Let
\[
 \Delta_h=\{(R_h(p),C_h(p)):0\le p<m\}.
\]
Both terminal families contain the same-leaf diagonal
$(R_h(p),C_h(p))$ as their third-row, third-column entry.
The leaf blocks are singletons, so this is one fixed physical cell
for every realization of the corresponding terminal tuple.
Lemma~\ref{lem:classification} therefore gives
$\tau_{H_4}(A_h)\le|\Delta_h|=m$.

For the reverse inequality, use offsets numbered from zero. For
each $p<m$, put $t_p=W^-_{h,p\bmod2}$ and choose a type-I copy
$Q_p$ whose two marker rows, two marker columns and dummy row
all have offset $p$ in their respective blocks. Use the singleton
leaf row and column with label $p$, and use offset $p\bmod2$
in the parent column block $C^-_{h-1}(\lfloor p/2\rfloor)$.
These choices are valid: the marker and dummy blocks have size $m$,
and each parent block has size two. They realize an ordered copy
by the type-I description and Lemma~\ref{lem:quotient}.

The copies $Q_p$ are pairwise disjoint on each axis. Marker
positions of different modes or groups lie in different blocks;
within any one such block, different labels $p$ give different
offsets. The dummy row offsets and leaf labels are also distinct.
Two selected parent columns could agree only if both
$\lfloor p/2\rfloor$ and $p\bmod2$ agree, which determines $p$.
The marker, leaf, dummy and parent categories used on a given axis
are mutually distinct. Thus the $m$ copies are, in particular,
pairwise cell-disjoint. Every set hitting all original copies must
meet each $Q_p$ in a different cell, and hence has size at least $m$.
Combining the bounds proves the transversal identity; the editing
identity is Theorem~\ref{thm:main}, and $m\ge2$ gives strictness.
\end{proof}

The strict gap between hitting original copies and producing a free
matrix sharpens the distinction already developed in \cite{OAI}.
A separate Lean supplement formalizes the exact transversal identity,
using terminal exhaustion from the primary classification theorem.

%% file: additions/parameters.tex
\section{Quantitative consequences of the exact formulas}\label{sec:quantitative}
The main theorem determines the obstruction at three scales: tree
depth $h$, matrix order $n_h$, and normalized edit distance
$\epsilon_h$.  Explicit changes of scale sharpen the comparison
with polynomial removal inequalities.

\begin{proposition}[Stretched-exponential copy-density profile]
\label{prop:stretched}
At the exact values $\epsilon_h=(14h+2)^{-2}$, the copy density is
\begin{equation}\label{eq:stretched}
 \delta_h=\biggl(\frac52-2^{-h}\biggr)\epsilon_h^4
    \exp\biggl[-\frac{\log 4}{14}
                   \bigl(\epsilon_h^{-1/2}-2\bigr)\biggr].
\end{equation}
In particular, the largest valid lower-bound function for the
individual pattern $H_4$ cannot exceed the right-hand side along
this sequence of edit parameters.
\end{proposition}
\begin{proof}
By Theorem~\ref{thm:main},
$\delta_h=(5/2-2^{-h})/((14h+2)^8 4^h)$.
The relation $\epsilon_h^{-1/2}=14h+2$ gives
$h=(\epsilon_h^{-1/2}-2)/14$. Substitution yields
\eqref{eq:stretched}. For any function $g(\epsilon)$ satisfying
$N_{H_4}(A)\ge g(\dist_{H_4}(A))n^8$ for all finite $A$,
applying the bound to $A_h$ gives $g(\epsilon_h)\le\delta_h$.
\end{proof}

\begin{proposition}[Profile by matrix order]\label{prop:byorder}
As $h\to\infty$, the constructed orders and copy counts satisfy
\[
 \epsilon_h=\Theta((\log n_h)^{-2}),\qquad
 \delta_h=\Theta\!\left(\frac1{n_h^2(\log n_h)^6}\right),
 \qquad
 N_{H_4}(A_h)=\Theta\!\left(\frac{n_h^6}{(\log n_h)^6}\right).
\]
\end{proposition}
\begin{proof}
Writing $d_h=14h+2$ and $m=2^h$ gives
$\log n_h=h\log2+\log d_h=\Theta(h)$, so
$d_h=\Theta(\log n_h)$ and $m=n_h/d_h$.
The exact copy formula may be rearranged as
\[
 N_{H_4}(A_h)=\biggl(\frac52-\frac1m\biggr)
                    \frac{n_h^6}{d_h^6}.
\]
The coefficient lies in $[2,5/2)$ because $m\ge2$.
Dividing by $n_h^8$ and using the distance identity gives the three
assertions.
\end{proof}

\subsection{A small-depth consistency table}
The exact formulas provide convenient integer-valued checks of the
terminal enumeration. They require no asymptotic approximation.
\begin{center}\small
\begin{tabular}{r|r|r|r|r|r}
\toprule
$h$ & $m$ & $n_h$ & family I & family II & $N_{H_4}(A_h)$\\\midrule
1&2&32&128&0&128\\
2&4&120&8192&1024&9216\\
3&8&352&524288&98304&622592\\
4&16&928&33554432&7340032&40894464\\
\bottomrule
\end{tabular}
\end{center}
The two contributions follow by substituting $m=2^h$ into the
form-I and form-II expressions. At depth one, family II is absent;
at higher depths, its share rises toward one fifth of the total.
The table illustrates why a search that detects only the single
explicit propagation family cannot certify the complete copy count.

%% file: additions/context.tex
\section{Background and interpretation}\label{sec:context}
The quantitative issue for ordered matrices is the dependence on the edit
parameter.  The row and column orders are independent: an occurrence of a
pattern $F\in\{0,1\}^{s\times t}$ is a pair of increasing injections
$[s]\hookrightarrow[n]$ and $[t]\hookrightarrow[n]$, and all $st$ entries,
including the zeros, are prescribed. The corresponding hereditary
property consists of matrices avoiding $F$ as an induced ordered
submatrix.  A removal bound supplies a lower bound on the number of
occurrences whenever one must change a prescribed fraction of all
entries to obtain a member of this hereditary property.

For an individual fixed $F$, a polynomial bound means the existence of
$c_F,C_F>0$ such that the right side of \eqref{eq:removal-bound} holds
for \emph{every} size and \emph{every} matrix.  Failure is witnessed by
one fixed pattern and a sequence of arbitrarily large hosts.  The
quantifiers matter: choosing a different forbidden pattern as the
accuracy parameter decreases would address a different question.
The present family uses the same $H_4$ at every depth; the pattern is
independent of $h$, $m$ and $n_h$.

We normalize occurrences by $n^{2k}$, as in
\eqref{eq:removal-bound}, and count increasing injections, so a copy
has no additional factorial multiplicity. For a fixed $k$, using
$\binom{n}{k}^2$ instead changes the density only by a factor bounded
above and below by positive constants for $n\ge k$. Thus the failure
of a polynomial bound is independent of these two common conventions.

Alon, Ben-Eliezer and Fischer \cite{ABF} proved qualitative testing
results for hereditary properties of ordered graphs and matrices.
Alon and Ben-Eliezer \cite{ABE20} obtained a polynomial dependence for
a \emph{hitting} formulation: a small set of entries may intersect all
copies present in the original matrix. Hitting the original copies
need not produce a free matrix after entry changes, because a changed
zero or one can complete a new induced copy.  The relevant full removal
problem minimizes changes over the resulting, genuinely $F$-free
matrices.  This distinction is particularly significant for an induced
pattern, whose constraints include prescribed zeros as well as ones.
Alon and Ben-Eliezer also establish a full polynomial removal lemma
for every fixed finite family of binary matrices closed under row
permutations \cite[Theorem~1.5]{ABE20}, and hence, by transposition,
for families closed under column permutations. The singleton
$\{H_4\}$ has neither closure property, so this positive result does
not apply to the present obstruction. Gishboliner and Shapira
\cite[Conjecture~4.5]{GS25} formulate the full polynomial removal
conjecture for a single fixed ordered binary matrix.
Towsner\'s ordered-hypergraph removal lemma provides a further
point of comparison at the level of qualitative removal \cite{Towsner}.
Earlier work on unordered bipartite matrices and conditional regularity
provided the polynomial-query comparison \cite{AFN07}; the ordering
requirement removes the row-and-column permutation freedom used in that
setting.

The 2026 construction of \cite{OAI} gave a fixed $66\times66$ binary
obstruction. Its tree-based construction showed that forcing gadgets can
make edit distance persist over successive levels, while an order
constraint can confine complete copies to terminal locations.  Our
pattern \eqref{eq:pattern} has four rows and four columns; its top-left
$2\times2$ anti-diagonal marker and the two identity coupling blocks
provide a smaller frame for the same kind of propagation.  The
sharper size is accompanied here by exact distance and copy-count
identities.  We isolate the roles of the order, the two signs, and
terminal identification in subsequent sections.

The tree-based propagation and terminal-sign identification build on
\cite{OAI}. The additional results here are the explicit $4\times4$
pattern, its exact physical copy count, its optimal unrestricted
repair and the formal certification of these identities.

\subsection{Three quantities that must be kept separate}
Let $\mathcal C_H(A)$ be the set of increasing row/column tuples
realizing $H$. Define the \emph{transversal number}
\[
 \tau_H(A)=\min\{|S|: S\subseteq[n]\times[n],\;
                 S\cap\operatorname{cells}(Q)\ne\varnothing
                 \ \forall Q\in\mathcal C_H(A)\}.
\]
Every edit set turning $A$ into an $H$-free matrix intersects every
copy originally in $A$; consequently
\begin{equation}\label{eq:hitting_le}
 \tau_H(A)\le D_H(A).
\end{equation}
Equality does not follow from the definition: a hitting set has no
obligation to avoid copies created by its own edits.  For our host,
the protected-selection argument establishes the stronger lower bound
$D_{H_4}(A_h)\ge m^2$ \emph{for all resulting matrices}.  The
canonical repair gives a matching upper bound by an explicit final
$H_4$-free matrix.  Neither step is supplied merely by counting
original occurrences.

Indeed, the terminal classification below shows a concrete gap between
the two quantities. The $m$ leaf-diagonal cells
$\{(R_h(p),C_h(p)):0\le p<m\}$ intersect every original copy,
so $\tau_{H_4}(A_h)\le m<m^2=D_{H_4}(A_h)$.
Corollary~\ref{cor:exacttransversal} proves that the hitting bound is exact.
Every terminal copy contains its same-leaf diagonal as its
third-row, third-column cell. This is the hitting-versus-removal
phenomenon emphasized in \cite{OAI}, here retained by the smaller
pattern; it is not a new general separation principle.

It is also useful to distinguish physical copies from type-level
copies. In a block-constant host, quotient rows and columns encode
signatures and ordering, whereas block multiplicities control the
physical number of realizations. The quotient alone has no absolute
copy count unless its tuples are weighted by those multiplicities.
This separation is essential in the terminal enumeration.

\subsection{The four-by-four pattern as a forcing frame}
Partition $H_4$ after the second row and the second column. The
four component blocks are respectively $J_2-I_2$, $I_2$, $I_2$,
and the asymmetric body $P$. Thus two marker rows and columns
fix twelve entries of a potential induced copy; the residual body
has only four entries. If the residual body has the form
$\left(\begin{smallmatrix}a&b\\c&d\end{smallmatrix}\right)$,
forbidding $H_4$ imposes the single constraint
\[
 (a,b,c,d)\ne(1,0,1,1).
\]
The construction uses this exact negative constraint in four
orientations.  Two orientations push a value from level $i-1$ to
a cross-level cell, and two push it from that cell to the diagonal
at level $i$.  The plus chain propagates the value one; the minus
chain propagates zero.  Their terminal positions coincide.  The
gadget constraints therefore become inconsistent unless an entry
required by a selected frame has been changed.  The XOR choice of
physical representatives makes those required entries disjoint
across $m^2$ choices.

%% file: additions/structure.tex
\section{Order and block structure in detail}\label{sec:structure}
This section records elementary consequences of the host definition
that are useful when interpreting the certificate classification.
The binary word of a node $(d,p)$ consists of the first $d$ bits in
the length-$h$ binary expansion of an appropriate leaf label.
An internal signed node is represented by the corresponding word,
together with its sign.  The leaf at depth $h$ has a single type.

\begin{lemma}[Tree comparison by words]\label{lem:wordorder}
Take two distinct variable blocks of the same axis. If their node
words have an incomparable prefix, their order is determined by the
first bit on which they differ. If one word is a proper prefix of
the other, the ancestor precedes the descendant precisely when its
sign is the opening sign for that axis. The opening signs are minus
for rows and plus for columns. If the two words agree, the blocks
are the two signs at one internal node, and the opening block precedes
the closing block.
\end{lemma}
\begin{proof}
At any internal vertex the recursive row list has the form
\[
 R_d^-(p),\quad \text{left subtree},\quad
 \text{right subtree},\quad R_d^+(p),
\]
whereas the column list exchanges the two signed blocks. If the
words separate into the two different child subtrees, the left
subtree precedes the right regardless of the signs within either.
If one word names an ancestor of the other, the ancestor's opening
block occurs before the descendant's entire subtree and its closing
block after it. Finally, equal words can name distinct blocks only
at an internal node with opposite signs; the displayed recursive
lists place the opening sign before the closing sign. These three
cases exhaust pairs of distinct variable blocks.
\end{proof}

\begin{proposition}[A uniform quotient principle]\label{prop:uniformquotient}
Let $F$ be an $s\times t$ binary matrix with no repeated row
signature and no repeated column signature. Let $M$ be any finite
block-constant ordered binary matrix whose row blocks and column
blocks are nonempty consecutive intervals. Then
\[
 N_F(M)=\sum_{(r,c)\in\mathcal Q_F(M)}
      \prod_{i=1}^{s}|R_{r_i}|\prod_{j=1}^{t}|C_{c_j}|,
\]
where $\mathcal Q_F(M)$ consists of the increasing row-block and
column-block tuples realizing the entries of $F$ on the quotient.
Moreover, $M$ is $F$-free if and only if its quotient has no such
tuple.
\end{proposition}
\begin{proof}
If two chosen physical rows belonged to the same block, then they
would have identical entries on every chosen column, contradicting
the distinctness of the corresponding rows of $F$. The column
argument is symmetric. It follows that each physical copy projects
to a unique quotient tuple with strictly increasing block indices.
Conversely, any tuple of quotient blocks realizing $F$ admits a
physical copy by choosing one position from each participating
block. Every such choice preserves the order, because the blocks
are intervals; it also preserves all entries, because entries are
block-constant. The choices are independent and their number is the
product of the participating block sizes. Summing over the quotient
tuples gives the identity and its freeness consequence.
\end{proof}

There are two independent hypotheses in Proposition~\ref{prop:uniformquotient}.
The consecutive-interval requirement is responsible for preservation
of the ordered embedding, while distinct row and column signatures
ensure that the selected type tuples are strictly increasing.
For patterns with repeated signatures one may have several physical
rows in the same type, and a multiplicity formula requires binomial
factors and within-block order. The simple product formula suffices
for $H_4$ and for the general twin-free patterns considered in
Appendix~\ref{app:smallwitness}.

\subsection{Exact support signatures}
For a mode $t$ and its two marker groups, write
$\rho_1(t),\rho_2(t)$ for row-role predicates and
$\kappa_1(t),\kappa_2(t)$ for column-role predicates.
The marker rows of group $g$ have entry one against a non-marker
column exactly when $\kappa_g(t)$ holds there. Similarly, marker
columns of group $g$ have entry one against a non-marker row
exactly when $\rho_g(t)$ holds.  The corner of a single mode is
$J_2-I_2$, while the corner between different modes is zero.
The parity index in the $W$ modes is therefore a genuine structural
datum. It separates the two candidate child branches using only
one bit, and it is recovered from the parity of the terminal leaf.

\begin{lemma}[Propagation frames realize the prescribed order]
\label{lem:framesorder}
At every depth $1\le i\le h$ and for every node $p_i$ with parent
$p_{i-1}$, the four row and column pairs displayed in
Section~\ref{sec:distance} are increasing tuples on their respective
axes. Their marker restrictions and their incidences with the dummy
blocks agree with the top twelve entries of $H_4$.
\end{lemma}
\begin{proof}
The first two selected positions are the two groups of one marker
mode, both before the variable and dummy positions. The order of
parent and child among variable rows is ancestor--descendant for
minus and descendant--ancestor for plus; among columns it is the
reverse. The first column in a $V$ frame is the column dummy,
which precedes the tree. The second row in a $W$ frame is the row
dummy, which follows it. Consequently the four increasing row and
column tuples are exactly those displayed in the forcing table.
The marker corner is the off-diagonal matrix $J_2-I_2$.
Membership in the two row roles gives the lower-left identity,
and membership in the two column roles gives the upper-right
identity. The dummy incidences are one by the entry definition.
The corresponding body entries may change under edits, yielding
precisely the implications shown in the forcing table.
\end{proof}

\subsection{Where the two signs meet}
At depths below $h$, plus and minus designate different physical
row and column blocks even when they carry the same node label.
Their diagonal entries can therefore consistently have opposite
values. At the leaf level there is only one block $R_h(p)$ and only
one block $C_h(p)$; the selected offsets are both zero. Thus the
positive and negative inductive requirements apply to the same
physical matrix cell, not simply to two equal labels. The
identification is the final step required by the edit-distance
argument. It also explains why the protected selection must use
one common leaf $z=a\oplus b$ on both axes.

%% file: additions/alternate_lower.tex
\section{A second proof of the editing lower bound}\label{sec:probmethod}
The exact disjointness of protected sets under the XOR selection is a
convenient deterministic proof. A probabilistic argument reaches the
same editing threshold using only one-coordinate uniformity. This
formulation shows which feature is intrinsic to the tree host and
which feature is special to the XOR coupling.
The marginal-exposure and union-bound method adapts the edit-distance
argument of \cite{OAI} to the smaller marker frame used here. It is
included to explain the construction and to compare it with the
deterministic packing, rather than as a new probabilistic method.

Choose a leaf label $Z$ uniformly from $\{0,\dots,m-1\}$. At every
depth $d$ select the node $\lfloor Z/2^{h-d}\rfloor$ and, in each
signed block at that node, sample one physical row and one physical
column uniformly from its $2^{h-d}$ offsets. Sample the row and
column representatives in the marker and dummy blocks uniformly
and independently of each other and of the tree representatives.
The two signed paths meet at the same selected physical leaf
because a leaf block consists of a single position.

\begin{lemma}[Single-position marginal]\label{lem:marginal}
For any fixed variable row position at depth $d$ and sign, its
probability of selection is $m^{-1}$. The same is true for a fixed
variable column position, marker position, or dummy position.
\end{lemma}
\begin{proof}
A depth-$d$ node is reached with probability $2^{-d}$.
Conditional on reaching it, a given offset is sampled with
probability $2^{-(h-d)}$. Multiplying gives $2^{-h}=m^{-1}$.
At the leaf depth the selection probability is $2^{-h}$ because
there is one representative per node. Marker and dummy blocks have
exactly $m$ positions and are sampled uniformly.
\end{proof}

Protect all selected cells for which at least one endpoint is a
marker or dummy position, together with the two selected
root-variable diagonal cells. No other variable--variable cells
are protected. In particular, the shared terminal cell, whose
value will be forced by the propagation, need not be protected.

\begin{lemma}[Two-coordinate exposure]\label{lem:twocoord}
For every fixed physical matrix cell $e$, the probability that $e$
is protected by the random selection is at most $m^{-2}$.
\end{lemma}
\begin{proof}
If either endpoint of $e$ is a marker or dummy position, that
endpoint is chosen independently of the tree-path and variable
offset choices on the other axis. Each endpoint has marginal
probability $m^{-1}$ by Lemma~\ref{lem:marginal}, and the two
events are independent, giving $m^{-2}$. If both endpoints are
variable positions at the root, their offsets are sampled
independently from the two root blocks of size $m$, giving the
same product. All other cells have protection probability zero.
\end{proof}

\begin{proposition}[Probabilistic lower bound]
\label{prop:problower}
Every edit set that turns $A_h$ into an $H_4$-free binary matrix
contains at least $m^2$ cells.
\end{proposition}
\begin{proof}
Let $E$ be an arbitrary set of changed cells with $|E|<m^2$.
By Lemma~\ref{lem:twocoord} and the union bound,
\[
 \mathbb P\bigl(E\cap\mathrm{Protected}\ne\varnothing\bigr)
 \le \sum_{e\in E}\mathbb P(e\in\mathrm{Protected})
 \le |E|m^{-2}<1.
\]
A selection therefore exists whose protected cells avoid every
edit. The forcing frames along its common path are valid and have
all twelve marker incidences unchanged. The same four implications
as in Section~\ref{sec:distance} propagate the root value one
through the plus chain and the root value zero through the minus
chain. Both implications concern the same selected terminal cell;
they are inconsistent in an $H_4$-free edited matrix. Thus $E$
cannot eliminate all copies.
\end{proof}

The probabilistic argument uses independence \emph{between axes}
whenever one position is in a marker or dummy block. It does not
require independence of the variable representatives on different
levels of the same axis, nor of the two terminal path labels.
The deterministic XOR proof goes further: each individual protected
cell identifies at most one of the $m^2$ selections. This is an exact
combinatorial packing statement, while the probabilistic proof
controls the probability of hitting each required cell. Both yield
the sharp integer threshold $m^2$.

\subsection{Boolean verification of the four implications}
The lower-right body $P$ is
$\left(\begin{smallmatrix}1&0\\1&1\end{smallmatrix}\right)$.
Each of the four forcing frames has three of its four body entries
fixed by the choice of mode or by the current inductive hypothesis.
The remaining entry is constrained by the absence of a completed
copy. The four exact incompatibilities are
\begin{align*}
 V_i^+:&\quad (x_{i-1},\alpha_i)=(1,0),\\
 V_i^-:&\quad (y_{i-1},\beta_i)=(0,1),\\
 W_{i,p_i\bmod2}^+:&\quad (\alpha_i,x_i)=(1,0),\\
 W_{i,p_i\bmod2}^-:&\quad (\beta_i,y_i)=(0,1).
\end{align*}
Each forbidden pair would complete the body $P$ and therefore
produce $H_4$ with the protected frame. The complement of each
pair yields exactly the implication in the forcing table. This
explicit truth-table reduction is independent of the arithmetic
classification certificates used for the upper bound.

%% file: additions/robust.tex
\section{Robustness of the edit obstruction under marker permutations}
\label{sec:robust}
Only a small part of the complete marker order enters the edit-distance
argument.  The original host places the two groups of each marker
mode consecutively and uses a particular ordering of all $6h$ modes.
The forcing argument itself requires only that the two groups of a
single mode occur in their given order and that the entire marker
prefix precede the variable and dummy blocks. This gives a useful
separation between the mechanism responsible for distance and the
mechanism responsible for terminal localization.

For each depth $h$, let $\pi_r$ and $\pi_c$ be two arbitrary
permutations of its $6h$ complete marker modes. Define
$A_h^{\pi_r,\pi_c}$ by permuting the marker \emph{row} blocks in
accordance with $\pi_r$ and marker \emph{column} blocks in
accordance with $\pi_c$, preserving the order of the two groups
inside each mode. Keep every block-constant entry, the two signed
tree orders, and the positions of the row and column dummy blocks
relative to the variable tree. The result is again a square matrix
of order $(14h+2)2^h$.

\begin{proposition}[Order-independent forcing lower bound]
\label{prop:permuted}
For every $h\ge1$ and every pair of complete-mode permutations,
\[
 D_{H_4}(A_h^{\pi_r,\pi_c})\ge 4^h.
\]
In particular, the lower bound holds against arbitrary binary edits
under independently reordered marker prefixes.
\end{proposition}
\begin{proof}
Use exactly the $m^2$ XOR-indexed selections from
Section~\ref{sec:distance}. The definition and injectivity of the
protected cells depend on physical block labels, their offsets,
the chosen tree path and the fact that marker and dummy blocks have
size $m$; permutation of whole marker modes changes none of these.
Therefore any edit set of size below $m^2$ misses the protected cells
of one selection.

For a fixed mode $t$, its two marker row groups appear in the order
$A_1(t)<A_2(t)$ and its marker column groups in the order
$B_1(t)<B_2(t)$. All four groups precede every variable and dummy
position, regardless of the position of other complete modes in the
prefix. The four forcing frames use only the two marker groups of
their own mode. Their increasing order and all the relevant
incidences therefore remain unchanged. The implications in the
forcing table again propagate $x_0=1$ and $y_0=0$ to the common leaf
cell, contradicting freeness. Thus any eliminating edit set has
size at least $m^2$.
\end{proof}

The upper bound $D_{H_4}(A_h)\le m^2$ and the exact copy count were
proved for the canonical mode order in Section~\ref{sec:host}.
Proposition~\ref{prop:permuted} records the broader range of the
forcing mechanism itself. The separately supplied Lean extension
contains the corresponding concrete-index theorem
\texttt{permuted\_matrix\_edit\_lower\_bound}; the primary
submission project retains its canonical order in the global
classification and removal statements.

\section{The selection map as an injective coding device}
\label{sec:selectionmap}
The XOR choice can be characterized without mentioning the particular
marker rules. Fix a word length $h$, and for each pair
$(a,b)\in\{0,\dots,2^h-1\}^2$ place a selected representative in
every signed variable block on the ancestral path of $a\oplus b$.
At a depth $d$, the selected node is the high-$d$-bit prefix of the
XOR, and the selected offset on the row axis (respectively column
axis) is the low $h-d$ bits of $a$ (respectively $b$).

\begin{lemma}[One-endpoint recovery]\label{lem:recovery}
Suppose the selected variable row at depth $d$ is known and the
selected marker-column offset is known. Then the pair $(a,b)$ is
uniquely determined. The analogous assertion holds with the two
axes reversed. A pair of selected root-variable positions also
uniquely determines $(a,b)$.
\end{lemma}
\begin{proof}
The marker-column offset equals $b$. The variable row determines
its node prefix $p_d$ and row offset $u=a\bmod2^{h-d}$. Since
bitwise XOR commutes with taking high-bit prefixes,
\[
 \lfloor a/2^{h-d}\rfloor
     =p_d\oplus\lfloor b/2^{h-d}\rfloor.
\]
The quotient on the left and remainder $u$ recover $a$. This
also holds for $d=h$, when the row offset is zero. Reversing the
axes recovers $b$ by the same equation. At $d=0$, the root node
is zero and the row and column offsets are exactly $a$ and $b$.
\end{proof}

This coding interpretation makes the protected-cell injectivity
transparent. Each protected cell contains at least one such
recoverable configuration; hence it belongs to at most one selected
path. The proof uses all $m^2$ labels simultaneously. Independent
choices of a random leaf and unrelated row and column offsets would
not provide this exact disjointness, even though some marginal
probabilities would still be uniform. The deterministic injection
also eliminates the need to estimate correlations among the
successive forcing gadgets.

%% file: additions/repair_spectrum.tex
\section{Which cells the canonical repair changes in an original copy}
\label{sec:repairspectrum}
The two terminal families give a precise description of the
intersection of each original copy with the canonical edit set
\[
 S_h=R_*\times\bigcup_{0\le p<m}C_h(p).
\]
Here $|S_h|=m^2$ is a count of physical cells. The classification
reveals that the same repair affects the two terminal forms with
different multiplicities.

\begin{proposition}[Repair-incidence spectrum]
\label{prop:repairspectrum}
Every type-I physical copy of $H_4$ in $A_h$ contains exactly one
cell from $S_h$. Every type-II copy contains exactly two cells from
$S_h$. The number of original copies meeting $S_h$ in one cell is
$2m^6$, and the number meeting it in two cells is
$\frac12(m-2)m^5$.
\end{proposition}
\begin{proof}
The unique row type belonging to $R_*$ in either terminal form is
the fourth row. In type I, the only chosen leaf column is the third
column $C_h(p)$; the two first columns are markers and the fourth
is the signed parent. Therefore precisely the fourth-row,
third-column physical cell belongs to $S_h$. In type II, both the
second and third columns are leaves, namely $C_h(q)$ and $C_h(p)$.
They are distinct since $q<p$. The first column is a marker and
the fourth is a signed parent. Thus exactly two cells of the copy
belong to $S_h$. The counting assertion follows from the exhaustive,
disjoint terminal classification and the weighted counts in
Section~\ref{sec:count}.
\end{proof}

\begin{corollary}[The canonical rectangle hits every old copy]
Every copy of $H_4$ present in $A_h$ has at least one entry changed
in the passage to $B_h$.
\end{corollary}
\begin{proof}
Both terminal families meet $S_h$ by
Proposition~\ref{prop:repairspectrum}, and there are no other
original copies.
\end{proof}

There are two distinct assertions in the proof of the sharp distance.
The corollary identifies a small set intersecting the originally
present copies; Lemma~\ref{lem:repair} proves that after all cells
of this set are changed simultaneously, the resulting physical
matrix is completely $H_4$-free. The latter requires the repaired
entry-table classification, because a cell edit can participate in
configurations absent from the original matrix. The arithmetic
repair certificate establishes freeness directly for the edited
matrix, rather than inferring it from the hitting corollary.

The spectrum also provides a quantitative accounting check.
If one sums, over original physical copies, the number of their
entries inside $S_h$, then the answer is
\begin{equation}\label{eq:repairincidence}
 2m^6+2\cdot\frac{m-2}{2}m^5
   =(3m-2)m^5.
\end{equation}
This number counts copy--edit-cell incidences with multiplicity;
it is deliberately different from both the number of copies
$\left(\frac52-\frac1m\right)m^6$ and the cardinality $m^2$
of the editing rectangle. Equation~\eqref{eq:repairincidence}
can be verified independently from the terminal parameterization.

%% file: formalization.tex
\section{Arithmetic certificates and formalization}\label{app:formalization}
This appendix specifies the integer model used in the classification and
repair proofs. Its domain and entry rules are uniform in $h$, so the two
certificate lemmas have a single depth parameter.

\subsection{The integer block domain}\label{app:integer}
For an integer $h\ge1$, the admissible records are of three kinds:
\begin{align*}
 M(j,i,g)&=(0,j,i,g,0,0,0),\\
 U(d,p,s)&=(1,0,0,0,d,p,s),\\
 D_*&=(2,0,0,0,0,0,0).
\end{align*}
The seven fields are, in order, kind, mode, marker level, marker group,
variable depth, node, and sign. A marker satisfies
$0\le j<6$, $1\le i\le h$ and $g\in\{1,2\}$. The mode codes are
\[
 0=V^+,\quad 1=V^-,\quad 2=W_0^+,\quad 3=W_1^+,
 \quad 4=W_0^-,\quad 5=W_1^-.
\]
A variable satisfies $0\le d\le h$, $p\ge0$, and $d=0\Rightarrow p=0$.
Its sign is $s=2$ when $d=h$, and $s\in\{0,1\}$ when $d<h$, with
$0,1,2$ denoting plus, minus and leaf, respectively.

The physical encoding assigns the appropriate mode code to a marker,
copies the depth, node and sign of a variable, and sends the dummy to
$D_*$. Thus every physical block gives an admissible record. The integer
domain relaxes physical validity by omitting the bound $p<2^d$.

\subsection{Comparison and entry predicates}
For markers, comparison is lexicographic on
\[
 \bigl(\ind[j<4],i,j,g\bigr).
\]
This places modes 4 and 5 first, in precisely the marker order of
Section~\ref{sec:host}. Every marker precedes every non-marker. Among
non-markers the dummy follows variables on rows and precedes them on
columns.

For variable records, write $(o_r,c_r)=(1,0)$ and
$(o_c,c_c)=(0,1)$ for the opening and closing signs on the two axes.
Put $q'=\lfloor q/2\rfloor$ and $p'=\lfloor p/2\rfloor$. The integer
comparison $\prec_a$, for $a\in\{r,c\}$, is defined by
\begin{equation}\label{eq:integer-order}
 U(d,p,s)\prec_a U(e,q,t)\ \Longleftrightarrow\
 \begin{cases}
 p<q\ \lor\ (p=q\land s=o_a\land t=c_a),&d=e,\\
 p<q'\ \lor\ (p=q'\land s=o_a),&e=d+1,\\
 p'<q\ \lor\ (p'=q\land t=c_a),&d=e+1,\\
 \mathrm{true},&|d-e|\ge2.
 \end{cases}
\end{equation}
At equal depths and adjacent depths these are the exact comparisons of
the tree blocks. At larger depth gaps the relation permits every pair.

Let $E_h^0$ be the original integer entry function. Its marker--marker
rule is
\[
 E_h^0(M(j,i,g),M(k,\ell,u))=\ind[j=k\land i=\ell\land g\ne u].
\]
The marker roles are those of the table in Section~\ref{sec:host}, with
level and parity tested on the records. In a signed role, $U(d,p,2)$
belongs to either sign when $d=h$. Entries with a dummy and a non-marker
are one. The variable entries are exactly \eqref{eq:values}, with the
depth and sign conventions above. Let $E_h^1$ have the same rules, except
\[
 E_h^1(D_*,U(h,p,2))=0\qquad(p\ge0).
\]
These definitions retain every entry rule while relaxing node ranges and
some comparisons.

For $\eta\in\{0,1\}$, define $\mathcal C_h^\eta(r,c)$ for two
four-tuples of records by the following conditions: all eight records are
admissible, $r_i\prec_r r_j$ and $c_i\prec_c c_j$ for every
$1\le i<j\le4$, and
\[
 E_h^\eta(r_i,c_j)=(H_4)_{ij}\qquad(1\le i,j\le4).
\]
In particular, all six comparisons on each axis are part of the copy
predicate.

\begin{lemma}[Encoding a physical copy]\label{lem:encoding}
The physical block encoding maps every ordered block copy in $A_h$ to
a witness for $\mathcal C_h^0$, and every ordered block copy in $B_h$ to
a witness for $\mathcal C_h^1$.
\end{lemma}
\begin{proof}
Physical marker levels and groups satisfy the integer ranges. For a
variable, physical validity gives its depth and sign, the nonnegative
node label, and the root condition. Thus the encoding preserves validity.

For comparisons, a marker pair has the same lexicographic key on both
axes. A marker--non-marker pair and a dummy--variable pair follow the
stated axis rules. For two variables at equal depth, node order and the
opening and closing signs give the first clause of
\eqref{eq:integer-order}. At adjacent depths the ancestor node is the
integer quotient by two, giving the second and third clauses. At a gap
of at least two levels the relaxed clause is satisfied. Hence all physical
comparisons imply their integer counterparts, separately for rows and
columns.

Finally, marker pairs use the same mode and group test, mixed marker
pairs use the same role test, and the non-marker entries use the same
dummy and variable rules. The canonical change is precisely on the
encoded dummy-row--leaf-column products. All sixteen entry equalities
are therefore preserved. These arguments are the encoding lemmas
\code{encode_valid}, \code{encode_before},
\code{encode_original_entry} and \code{encode_repaired_entry}; together
they give \code{encode_copy}.
\end{proof}

\subsection{The certificate statements}
Define $T_h(r,c)$ to mean that there exist $j\in\{4,5\}$ and $p\ge0$
such that
\begin{align*}
 r&=\bigl(M(j,h,1),M(j,h,2),U(h,p,2),D_*\bigr),\\
 c&=\bigl(M(j,h,1),v,U(h,p,2),U(h-1,\lfloor p/2\rfloor,1)\bigr),
\end{align*}
where either $v=M(j,h,2)$, or $v=U(h,q,2)$ with
$0\le q<p$ and $q\equiv p\pmod2$. The marker parity is recovered
from the entry conditions when this predicate is decoded in the physical
host.

\begin{lemma}[Arithmetic classification and repair]\label{lem:arithmetic}
For every integer $h\ge1$ and every pair of four-tuples $r,c$ of integer
records,
\begin{equation}\label{eq:certificate-claims}
 \mathcal C_h^0(r,c)\ \Longrightarrow\ T_h(r,c),
 \qquad
 \mathcal C_h^1(r,c)\ \Longrightarrow\ \bot.
\end{equation}
\end{lemma}
\begin{proof}[Certificate proof]
The first implication is the Lean theorem
\code{Arithmetic.classification}, and the second is
\code{Arithmetic.repaired_no_copy}. The classification derivation has
266 source segments and the repair derivation has 243. Their steps are
propositional inferences, linear integer arithmetic lemmas, and
congruence proofs. The transparent formula abbreviations express the
admissibility, comparison and entry predicates specified above.

The two bridge modules instantiate the final derivations with the fields
of $r$ and $c$. Both copy predicates supply eight validity conditions,
twelve pairwise comparisons and all sixteen entry conditions. The
classification derivation needs only fifteen of the entry conditions:
it does not use the condition at $(0,0)$ (with zero-based indices).
The repair derivation uses all sixteen. For classification the negation
of $T_h$ gives the counterexample premise; for repair the premises are
the repaired copy conditions themselves. Thus the omitted classification
premise strengthens that certificate without weakening the definition
of an induced copy.
The resulting contradictions establish \eqref{eq:certificate-claims}.
The derivations and the bridge proofs are checked by the Lean kernel.
\end{proof}

To decode the first conclusion, the leaf-row entry against the first
marker column is one. The role test consequently gives
$j=4+(p\bmod2)$. Physical validity then supplies $p<2^h$ and the valid
minus parent. This is the theorem \code{terminal_decodes}, which yields
the terminal families in Section~\ref{sec:classification}.

\subsection{Formal statements and verification}
The matrix family \code{standardHost} uses the usual order on
$\operatorname{Fin}(n_h)$ on both axes. Its copy predicate
uses increasing four-tuples and all sixteen Boolean equalities. Its
edit distance is the attained minimum Hamming cost over pattern-free
binary matrices. The principal mathematical interfaces are
\begin{center}\small
\begin{tabular}{p{47mm}p{99mm}}\toprule
Result&Lean declaration\\\midrule
Physical size and order&\code{physicalPairs_length}, \code{orderedPosition_order_iff}\\
Matrix copy bridge&\code{standard_copy_iff}\\
Terminal classification&\code{classification_iff}\\
Repaired freeness&\code{canonical_standard_free}\\
Exact copy count&\code{exact_copy_count}\\
Arbitrary-edit lower bound&\code{arbitrary_edit_lower_bound}\\
Exact edit distance&\code{exact_edit_distance}, \code{exact_normalized_distance}\\
Asymptotic ratio&\code{normalized_copy_ratio_tendsto}\\
Removal obstruction&\code{no_polynomial_removal_bound}\\
Main theorem&\code{end_to_end_main_theorem}\\\bottomrule
\end{tabular}
\end{center}

The XOR construction is developed in \code{XorSampling.lean}; its counting
argument is in \code{FiniteAvoidance.lean}. Terminal parameters are counted
in \code{TerminalCount.lean}, and their physical realizations are
enumerated in \code{CopyFibres.lean}. The module \code{GlobalCounting.lean}
identifies their disjoint union with the full copy enumeration.
\code{EditDistance.lean} transports Hamming cost through the
order-preserving bijections and proves attainment of the minimum.

The project pins Lean 4.24.0 and Mathlib at revision
\begin{center}\code{f897ebcf72cd16f89ab4577d0c826cd14afaafc7}.\end{center}
The clean verification procedure is run from the project root by
\begin{quote}\ttfamily\small
bash verify\_all.sh
\end{quote}
Its declaration audit selects every definition and theorem by its defining
project module, including private and anonymous declarations, and checks
the complete axiom dependency closure. Only \code{propext},
\code{Classical.choice} and \code{Quot.sound} are allowed. Current run
records, exact source and artifact hashes, and dependency provenance are
under \code{evidence/strict_build/}; supplied older logs are preserved
separately as historical evidence. Appendix~\ref{app:reproduce} describes
the build procedure, proof dependencies and retained verification records.

%% file: additions/compression.tex
\section{A uniform finite-witness principle for the host family}
\label{app:smallwitness}
The integer certificate in Appendix~\ref{app:formalization} handles
arbitrary depths directly. The same host also admits a finite-depth
reduction that is useful for independent specification checks.
We give its combinatorial proof in the notation of the present paper.
The statement applies to any square pattern whose row signatures and
column signatures are individually distinct. This structural result
is logically separate from the arithmetic certificate proof of
Theorem~\ref{thm:main}.

\begin{lemma}[Number of binary branching levels]\label{lem:branchlevels}
For a family of $v\ge1$ finite binary words, at most $v-1$ different
bit positions are the first divergence position of some pair before
either word terminates.
\end{lemma}
\begin{proof}
Form the trie of the given words, regarding a word that is a proper
prefix of another as a marked internal vertex. A first divergence
occurs only at a vertex with two nonempty child branches. A rooted
tree with $v$ distinct marked word endpoints has at most $v-1$ such
branching vertices: each branching vertex separates at least two
components containing marked endpoints, and contracting all degree-one
vertices yields a rooted tree with at most $v$ leaves after marked
internal vertices are represented as terminal branches. Consequently
the number of distinct branching depths is at most $v-1$.
\end{proof}

\begin{lemma}[Retained-level preservation]\label{lem:retention}
Choose at most $2k$ row or column quotient types from the host at
height $h$. Let $S\subseteq\{0,\dots,h\}$ contain $0,h$, the
defining depth of each selected marker or variable type, the
immediately preceding depth of each such positive defining depth,
and every first-divergence level of a pair of selected variable
words. Write $S=\{s_0=0<s_1<\cdots<s_{h'}=h\}$. Replace every
selected defining depth $s_i$ by $i$ and every selected node word
by its bits in retained positive positions. Then the compressed
types preserve all required block comparisons and entries of both
the original host and its canonical repair.
\end{lemma}
\begin{proof}
The rank map preserves comparisons and equality of selected depths.
The predecessor of every selected positive depth is retained.
Therefore two selected depths are consecutive precisely when their
images are consecutive; a gap of length at least two retains a point
strictly between the endpoints. The first-divergence bit of every
pair is retained, so lexicographic order and the proper-prefix
relation are preserved. Since the final bit at any defining variable
depth is retained, so is its parity. Parent comparison at adjacent
depths is preserved because the relevant predecessor bit is retained.

These are exactly the relations appearing in the host rules:
marker pairs test equality of mode and inequality of group;
marker--variable pairs test family, sign, depth adjacency, and
parity; variable pairs test equal or adjacent depths with their
lexicographic or parent comparison. The dummy incidence depends
only on type; the repair tests dummy row against a terminal leaf
column. Every selected type survives, its sign and marker group are
unchanged, and these Boolean tests have the same values after
compression. Lemma~\ref{lem:wordorder} gives preservation of the
signed tree order. The marker mode order depends on family, retained
depth and group, so its relevant comparisons are preserved as well.
For equal node words at an internal level, the two signs remain
distinct and their opening-before-closing order is unchanged.
\end{proof}

\begin{theorem}[General $6k$-level small-witness theorem]
\label{thm:sixk}
Fix $k\ge1$ and a binary $k\times k$ matrix $F$ with pairwise
distinct rows and pairwise distinct columns. If $F$ occurs as an
induced ordered copy in $A_h$ or $B_h$ for some $h\ge1$, then it
occurs in the corresponding host at some depth
$1\le d\le\min\{h,6k\}$. Consequently, the whole original
(resp. repaired) host family is $F$-free if and only if its depths
$1,\dots,6k$ are all $F$-free.
\end{theorem}
\begin{proof}
By Proposition~\ref{prop:uniformquotient}, the given physical copy
projects to $2k$ types, distinct among the selected row types and
among the selected column types. Retain their defining levels and
predecessors, as well as $0,h$ and every first-divergence bit level
of the selected variable words. If $v\ge1$ is the number of
variable types, Lemma~\ref{lem:branchlevels} bounds the number of
first-divergence levels by $v-1\le2k-1$. Thus
\[
 |S|\le 2+2k+2k+(2k-1)=6k+1.
\]
If $v=0$, the still stronger bound $|S|\le2+4k$ holds. Since both
endpoints $0,h$ are in $S$, the reduced height
$d=|S|-1$ satisfies $1\le d\le\min\{h,6k\}$.
By Lemma~\ref{lem:retention}, the compressed types retain each
of the $k^2$ prescribed entries and the two increasing orders.
They remain distinct within either axis: different defining kinds,
levels or signs remain different, and distinct node words remain
different because their first divergence is retained. Every
compressed block has positive multiplicity, so representatives of
the new types furnish a physical copy of $F$ at depth $d$.
The converse direction of the final equivalence is immediate;
if some larger depth contained a copy, the first assertion would
supply a copy at one of the finitely many depths.
\end{proof}

\begin{corollary}[Twenty-four levels suffice for a copy witness]
\label{cor:twentyfour}
Any $H_4$ copy in $A_h$ or $B_h$ at arbitrary depth yields an
$H_4$ copy in a host of the same kind at depth at most $24$.
\end{corollary}
\begin{proof}
Apply Theorem~\ref{thm:sixk} with $k=4$.
\end{proof}

\begin{remark}[Proof scope]
The finite-witness theorem is available in the separately supplied
Lean development in the declarations
\code{physical_pattern_copy_has_small_witness}
and \code{matrix_pattern_copy_has_small_witness}. These exported
interfaces certify $1\le d\le6k$. The additional bound $d\le h$
in the paper statement follows from $S\subseteq\{0,\dots,h\}$ in
the displayed combinatorial construction; it is not part of those
exported Lean conclusions.
The main theorem of this article uses the arithmetic-classification
and arithmetic-repair route in the primary formal source release.
\end{remark}

\subsection{Finite decision and its precise meaning}
Theorem~\ref{thm:sixk} supplies an exact finite decision procedure
for the assertion that \emph{every depth} of the specified host family
is $F$-free. For each of the $6k$ allowed depths, enumerate ordered
quotient tuples with $k$ distinct row and $k$ distinct column types
and evaluate all $k^2$ entries. Absence at all these depths implies
absence in every larger host. The statement concerns a fixed explicit
family of hosts, rather than arbitrary $n\times n$ matrices; its
finite reduction does not impose a polynomial running-time bound in
$k$.  For $H_4$, the arithmetic approach provides a direct all-depth
proof without the need to run the finite enumerator.

%% file: additions/verify.tex
\section{Reproducible verification and the proof interfaces}
\label{app:reproduce}
The formalization has four components: physical block positions and their
multiplicities; the ordered quotient host; arithmetic classification and
repaired freeness; and transfer to matrices indexed by
$\operatorname{Fin}(n_h)$. The theorem
\code{end_to_end_main_theorem} combines these components with the
exact-count, arbitrary-edit distance and asymptotic results.

\subsection{Dependencies of the principal statements}
\begin{center}\small
\begin{tabular}{p{43mm}p{54mm}p{37mm}}
\toprule
Mathematical assertion&Lean declaration or module&Role\\\midrule
Arithmetic original classification&\code{Arithmetic.classification}&Copy exhaustion\\
Repaired freeness&\code{Arithmetic.repaired_no_copy}&No new copies\\
Encoding of physical copies&\code{Arithmetic.encode_copy}&Arithmetic soundness\\
Terminal decoding&\code{terminal_decodes}&Physical terminal forms\\
Numeric matrix order&\code{orderedPosition_order_iff}&Order transfer\\
Protected sampling&\code{XorSampling.lean}&Arbitrary edits\\
Finite avoidance&\code{FiniteAvoidance.lean}&Untouched witness\\
Copy multiplicity&\code{CopyFibres}, \code{TerminalCount}&Exact count\\
Edit distance&\code{EditDistance.lean}&Exact minimum\\
Asymptotics&\code{Limits.lean}&Removal failure\\
\bottomrule
\end{tabular}
\end{center}

The arithmetic classification and repair derivations contain 266 and
243 source segments, respectively. The encoding checks all eight
validity conditions, twelve axis-order comparisons and sixteen pattern
entries. The quotient-to-physical correspondence and the numeric-order
bijections complete the passage to the ordinary indexed matrix.
The proof segments and every imported Lean module are distributed as
source files in the accompanying verification archive.

\subsection{Rebuild and axiom audit}
The project uses Lean 4.24.0 and pins Mathlib to revision
\begin{center}\code{f897ebcf72cd16f89ab4577d0c826cd14afaafc7}.\end{center}
From the \code{formal_verification/} directory, execute
\begin{quote}\ttfamily\small
bash verify\_all.sh
\end{quote}
The script retrieves the pinned dependency cache, builds the primary
library, checks the main theorem and its axiom audit, then rebuilds the
isolated structural-extension and exact-transversal projects. The main
library is complete with all arithmetic certificates, which are needed
for the terminal classification and canonical freeness results.
The compiled dependency audits permit only \code{propext},
\code{Classical.choice} and \code{Quot.sound}.

The archived primary build record reports a compilation of
550 project modules, with 116,338 declarations (93,189 theorems)
in the strict declaration audit. The independent extension-core
build record covers 34 library modules and 1,570 declarations
(991 theorems). The exact-transversal supplement records 49 declarations,
including 32 theorems, and derives terminal exhaustion from the primary
global classification theorem. The corresponding run reports and source
identities are retained under \code{evidence/} and in the supplemental
project directories. The archived records identify the source and environment of earlier
compilations; the command above performs a new build on the recipient's machine.

\subsection{Mathematical cross-checks}
A direct check of the construction first verifies block-constant entries
and the four increasing forcing frames. The encoding implication then
maps every physical copy to the arithmetic host. The classification
certificate excludes all nonterminal copies; the repair certificate
excludes every copy after the canonical edit. Finally, weighted terminal
multiplicities give the exact number of original copies. The main
Lean theorem assembles these facts with the arbitrary-edit lower bound
and the asymptotic comparison.